\documentclass[journal]{IEEEtran}
\pdfoutput=1

\usepackage{cite}

\usepackage[english]{babel}
\usepackage[cmex10]{amsmath}
\usepackage{amssymb,amsfonts}
\usepackage{multirow}
\usepackage{amsthm}
\usepackage{microtype}
\usepackage{caption}
\usepackage{graphicx}
\usepackage{etoolbox}
\usepackage[caption=false,font=footnotesize]{subfig}
\usepackage{pgfplots}
\usepackage{pgfplotstable}

\pgfplotsset{compat=newest,
    /pgfplots/ybar legend/.style={
    /pgfplots/legend image code/.code={%
       \draw[##1,/tikz/.cd,yshift=-0.25em]
        (0cm,0cm) rectangle (3pt,0.8em);},
   },
}

\newtoggle{isreport}
\toggletrue{isreport}

\usepackage[normalem]{ulem}

\newtheorem{lemma}{Lemma}

\newtheorem{prop}[lemma]{Proposition}

\theoremstyle{remark}
\theoremstyle{theorem}
\theoremstyle{theorem}\newtheorem{defn}[lemma]{Definition}

\newcommand{\paren}[1]{\big(#1\big)}

\newcommand{\set}[1]{\left\{#1\right\}}
\newcommand{\abs}[1]{\left|#1\right|}

\newcommand{\normm}[1]{{\left\vert\kern-0.25ex\left\vert\kern-0.25ex\left\vert #1
    \right\vert\kern-0.25ex\right\vert\kern-0.25ex\right\vert}}

\DeclareMathOperator{\spann}{span}

\DeclareMathOperator{\diag}{diag}

\newcommand{\R}{\mathbb R}

\newcommand{\eps}{\epsilon}
\newcommand{\lam}{\lambda}

\newcommand{\ul}[1]{\underline{#1}}
\newcommand{\ol}[1]{\overline{#1}}

\newcommand{\goesto}{\rightarrow}

\newcommand{\bs}{\backslash}

\newcommand{\calB}{\mathcal{B}}
\newcommand{\calC}{\mathcal{C}}
\newcommand{\calD}{\mathcal{D}}
\newcommand{\calE}{\mathcal{E}}
\newcommand{\calF}{\mathcal{F}}
\newcommand{\calG}{\mathcal{G}}

\newcommand{\calL}{\mathcal{L}}

\newcommand{\calN}{\mathcal{N}}

\newcommand{\calP}{\mathcal{P}}

\newcommand{\calR}{\mathcal{R}}

\newcommand{\n}{\mathfrak{n}}

\newcommand{\bff}[1]{{\bf #1}}

\usepackage{color}
\usepackage{enumerate}

\definecolor{myred}{RGB}{202,0,32}
\definecolor{myorange}{RGB}{244,165,130}
\definecolor{myviolet}{RGB}{194,165,207}
\definecolor{mycyan}{RGB}{146,197,222}
\definecolor{myblue}{RGB}{5,113,176}
\definecolor{mygreen}{RGB}{127,191,123}
\definecolor{mytile}{RGB}{27,120,55}

\IEEEoverridecommandlockouts

\title{Adaptive Network Response to Line Failures in Power Systems}

\author{Chen~Liang, Linqi~Guo, Alessandro~Zocca, Steven H.~Low,~and~Adam~Wierman
\thanks{This work has been supported by Resnick Fellowship, Linde Institute Research Award, NWO Rubicon grant 680.50.1529, 
NSF through grants CCF 1637598, ECCS 1619352, ECCS 1931662, CNS 1545096, CNS 1518941, CPS ECCS 1739355, CPS 154471.}
\thanks{CL, LG, SHL, AW are with the Department of Computing and Mathematical Sciences, California Institute of Technology, Pasadena,
CA, 91125, USA. Email: \texttt{\{cliang2, lguo, slow, adamw\}@caltech.edu}. AZ is with the Department of Mathematics of the Vrije Universiteit Amsterdam, 1081HV, Netherlands. Email: \texttt{a.zocca@vu.nl}.}}

\begin{document}

\maketitle
\thispagestyle{empty}
\pagestyle{empty}

\begin{abstract}
Transmission line failures in power systems propagate and cascade non-locally. In this work, we propose an adaptive control strategy that offers strong guarantees in both the mitigation and localization of line failures. Specifically, we leverage the properties of network \emph{bridge-block decomposition} and a frequency regulation method called the \emph{unified control}. If the balancing areas over which the unified control operates coincide with the bridge-blocks of the network, the proposed strategy drives the post-contingency system to a steady state where the impact of initial line outages is localized within the areas where they occurred whenever possible, stopping the cascading process. When the initial line outages cannot be localized, the proposed control strategy provides a configurable design that progressively involves and coordinates more balancing areas.
We compare the proposed control strategy with the classical Automatic Generation Control (AGC) on the IEEE 118-bus and 2736-bus test networks. Simulation results show that our strategy greatly improves overall reliability in terms of the $N-k$ security standard, and localizes the impact of initial failures in the majority of the simulated contingencies. Moreover, the proposed framework incurs significantly less load loss, if any, compared to AGC, in all our case studies.
\end{abstract}

\section{Introduction}\label{section:intro}
Transmission line failures in power systems propagate both locally and non-locally, making it challenging to design control methods that reliably prevent and control line failure cascades in power networks~\cite{bienstock2007integer,hines2007controlling}. Current industry practice for failure prevention and mitigation mostly relies on $N-1$ security constrained OPF as well as simulation-based contingency analysis~\cite{baldick2008initial}, which are often implemented in tertiary control and operates on a slow timescale. As a result, such control usually prescribes a conservative operation point that strictly prohibits line overloads in all possible scenarios.

Since tertiary control can take effect 5 minutes to more than an hour after the disturbance, most literature on cascading failure analysis adopts the assumption that injections remain unchanged after a line failure if the post contingency network remains connected (and are changed only after a bridge failure according to a generic balancing rule that re-balances power in each island)~\cite{bienstock2007integer,bienstock2010n-k,bernstein2014power,soltan2015analysis}. This assumption, however, is unrealistic and tends to be pessimistic: it does not take into account of frequency control mechanisms that adjust injections of controllable generators and loads immediately in response to line outages, on a faster timescale than that of post-contingency line tripping. 

In this paper, we augment the existing steady state cascading failure models with frequency control dynamics that affect power flow redistribution in the new equilibrium post contingency. This integrated failure model is not only more realistic, but also offers additional means to mitigate cascading failure through better design of the frequency control mechanism. Our proposed control strategy builds upon this extra freedom and reacts to line outages on a timescale of minutes.

\textbf{Contributions of this paper:} \emph{We integrate a distributed frequency control strategy with a tree-partitioned network to provide provable failure mitigation and localization guarantees on line failures.} This strategy operates on a different timescale and supplements current practice, improving both grid reliability and operation efficiency. To the best of our knowledge, this is the first attempt to leverage results from the frequency regulation literature in the context of cascading failures, bringing new perspectives and insights to both literature. Our proposed strategy guarantees that (a) whenever it is feasible to avoid it, line failures do not propagate, and (b) the impact of line failures is localized as much as possible in a manner configurable by the system operator. A preliminary version of this work is presented in~\cite{guo2019less}. In this paper we include detailed proofs that were omitted in~\cite{guo2019less} and extend our simulations to further illustrate the performance of our approach.

We introduce the main idea of this new mitigation strategy in Section \ref{section:control_strategy}, which makes use of the so-called \textit{Unified Controller (UC)}, a recent mechanism developed in the frequency regulation literature \cite{zhao2014design,li2016connecting,zhao2016unified,mallada2017optimal,zhao2018distributed}. 
We specifically leverage the ability of UC to enforce line limits on a faster timescale than thermal line failure dynamics whenever possible.
Our design revolves around the properties that emerge when the balancing areas that UC manages are connected in a tree structure. More specifically, in Section \ref{section:non_critical}, we characterize how UC responds to an initial failure, and prove that any \textit{non-critical failure} is automatically mitigated and localized. Later, in Section~\ref{section:critical}, we discuss how the system operator can configure its mitigation strategy to minimize the impact of \textit{critical failures}, and show that UC can be extended to detect such scenarios as part of its normal operation.

In order to establish these results, we propose an integrated failure propagation model in Section~\ref{section:failure_model}, which lies at the interface between the fast-timescale of the frequency dynamics and the slow-timescale of the line tripping process. We further prove new results on the UC optimization problem with the spectral representation of DC power flow equations. Lastly, we apply classical results from convex optimization to show that critical failures can be detected in a distributed fashion.

In Section \ref{section:case_studies}, we compare the proposed control strategy with 
classical Automatic Generation Control (AGC) using the IEEE 118-bus and 2736-bus test networks. We demonstrate that by switching off only a small number of transmission lines and adopting UC as the frequency controller, one can significantly improve overall grid reliability in terms of the $N-k$ security standard. Moreover, in a majority of the load profiles that are examined, our control strategy localizes the impact of initial failures to the balancing area where they occur, leaving the operating points of all other areas unchanged. This decoupling across balancing areas is important in practice. Lastly, we highlight that when load shedding is necessary, the proposed strategy incurs significantly smaller load loss.

Load shedding as a mitigation approach for cascading failures has been studied in the literature, e.g., 
using reinforcement learning \cite{Jung2002}, by shedding only the interruptible part of each load \cite{Faranda2007}, by adaptively using affine control based on observed states \cite{bienstock2011control},
by formulating it as a DC OPF problem \cite{Pahwa2013}, or by treating the cascading process as a discrete-time optimal control problem \cite{BaSavla2019}. In this paper we use the same DC power flow model as~\cite{bienstock2011control, Pahwa2013, BaSavla2019}, but, unlike these works, we do not propose separate load shedding schemes to be implemented on a slow (power flow) timescale. Instead, load and generation control under the UC framework is performed as part of frequency regulation, during both normal operation and a contingency.



\section{An Integrated Failure Propagation Model}\label{section:failure_model}

In this section, we present an integrated cascading failure model that incorporates frequency control dynamics and generalizes the steady-state DC failure model. 
%
\subsection{Fast-timescale Dynamics}\label{section:model}
Consider a power transmission network described by a graph $\calG=(\calN,\calE)$, where $\calN$ is the set of buses and $\calE \subset \calN\times \calN$ are transmission lines. Using the notation in Table~\ref{table:notations}, the post-contingency linearized frequency dynamics are:
\begin{subequations}\label{eqn:swing_and_network_dynamics}
\begin{IEEEeqnarray}{rCll}
\dot\theta_j & = & \omega_j, & \quad  j\in\calN \label{eqn:swing_dynamics.1}\\ 
  M_j\dot{\omega}_j &=& r_j + d_j - D_j\omega_j - \sum_{e\in\calE}C_{je}f_e,&\quad j\in\calN \label{eqn:swing_dynamics}\\
  {f}_{ij}&=&B_{ij}(\theta_i-\theta_j),\quad& (i,j)\in\calE. \label{eqn:network_flow_dynamics}
\end{IEEEeqnarray}
\end{subequations}
The above differential equations model the fast-timescale response of the system to a transmission line failure. Note that we assume generator voltage control is local and converges at a faster timescale than generator-load frequency control.  Therefore voltages are assumed to be fixed at their nominal values at the frequency control timescale.

The post-contingency injection deviation $p_j(t):=r_j + d_j(t)$ is the sum of the post-contingency disruption $r_j$ and the system response $d_j(t)$. 
The disruption $r:=(r_j, j\in\calN)$ comes from the effect of transmission line failures. In particular, suppose a failure of line $(s,t)$ with pre-contingency flow $f_{st}^\text{pre}$ happens, the post-contingency disruption is $r_s=f_{st}^\text{pre}$, $r_t = -f_{st}^\text{pre}$ and $r_j=0$ for all other entries. 
The vector $d(t):=(d_j(t), j\in\calN)$ models frequency control and their values are 
determined by a feedback control mechanism (a non-controllable constant-power load is simply a special case where the controls are set as $d_j(t) \equiv d_j$). We assume in this paper that the feedback controller is stabilizing and drives the closed-loop system towards an equilibrium as long as the post-contingency disruption $r$ can be feasibly mitigated (see Section \ref{section:critical} for more discussion).

\begin{table}[t]
\def\arraystretch{1.1}
\caption{Variables associated with buses and transmission lines.}\label{table:notations}
\vspace{.2cm}

\begin{tabular}{|m{2.7cm}|m{5cm}|}
\hline
$\theta:=(\theta_j, j\in\calN)$ & bus voltage angle deviations from pre-contingency values\\
\hline
$\omega:=(\omega_j, j\in\calN)$ & bus frequency deviations from pre-contingency values\\
\hline
$r:=(r_j, j\in\calN)$ & system disturbances\\
\hline
$d:=(d_j, j\in\calN)$ & power injection/controllable load deviations from pre-contingency values for generator/load buses\\
\hline
$p:=(p_j, j\in\calN)$ & aggregate post-contingency injection deviation\\
\hline
$\ol{d}_j, \ul{d}_j, j\in\calN$ & upper and lower limits for the adjustable injection $d_j$\\
\hline
$D_j\omega_j, j\in\calN$ & aggregate generator damping for generator buses; aggregate load frequency
response for load buses\\
\hline
$M_j, j\in\calN$ & inertia constants\\
\hline
$f:=(f_e, e\in\calE)$ & branch flow deviations from pre-contingency values\\
\hline
$\ol{f}_e, \ul{f}_e, e\in\calE$ & upper and lower limits for branch flow deviations\\
\hline
$n:=|\calN|$ & number of buses \\
\hline
$m:=|\calE|$ & number of transmission lines \\
\hline
$C\in \R^{n\times m}$ & post-contingency incidence matrix of $\calG$: $C_{je}=1$ if $j$ is the source of $e$, $C_{je}=-1$ if $j$ is the destination of $e$, and $C_{je}=0$ otherwise\\
\hline
$B:=\diag(B_e, e\in\calE)$ & branch flow linearization coefficients that depend on line susceptances, nominal voltage magnitudes and reference phase angles \\
\hline
\end{tabular}
\vspace{-.2cm}
\end{table}

\begin{defn}
A state $x^*:=(\theta^*, \omega^*, d^*, f^*)\in\R^{3n+m}$ is said to be a \textbf{\emph{closed-loop equilibrium}} or
simply an \textbf{\emph{equilibrium}} of \eqref{eqn:swing_and_network_dynamics} if the right hand sides of \eqref{eqn:swing_dynamics.1}\eqref{eqn:swing_dynamics} are zero and \eqref{eqn:network_flow_dynamics} is satisfied at $x^*$.
\end{defn}
The frequency dynamics \eqref{eqn:swing_and_network_dynamics} implies that any equilibrium configuration $x^*$ satisfies
\begin{IEEEeqnarray*}{C}
w^* = 0, \quad  p^* =  r+d^* = Cf^*, \quad  f^* = BC^T\theta^*,
\end{IEEEeqnarray*}
In other words, $x^*$ satisfies the DC power flow model.\footnote{In primary frequency control literature (see \cite{zhao2014design,zhao2016unified} for instance), the right hand side of \eqref{eqn:swing_dynamics.1} is not required to be zero for an equilibrium point $x^*$. We impose this requirement on \eqref{eqn:swing_dynamics.1} here as our discussion focuses on controllers that achieve secondary frequency control and thus $\omega^*=0$ always holds. Our model and results can be readily extended to the case where $\omega^*\neq 0$; see Appendix \ref{section:previous_model} for more details.}

The equilibrium to which the closed-loop system \eqref{eqn:swing_and_network_dynamics} converges thus determines the post-contingency DC power flow solution and can in turn impact how failures propagate in the network. A key insight from~\cite{zhao2014design,li2016connecting,zhao2016unified,mallada2017optimal,zhao2018distributed} is that the closed-loop equilibrium $x^*$ of \eqref{eqn:swing_and_network_dynamics} is also an optimal solution of a certain DC-based optimization problem that can be determined explicitly.  
Different frequency controllers $d(t)$ induce different dynamics \eqref{eqn:swing_and_network_dynamics}, whose closed-loop equilibria solve optimization problems with corresponding objective functions and constraints. As such, different frequency controllers can alternatively be modelled by the underlying optimization problems that their equilibria solve. 

We remark that the model \eqref{eqn:swing_and_network_dynamics} can be extended to include load buses $j$ where $M_j=0$.  This system of differential algebraic equations assumes that the load dynamics evolve at a faster timescale than generator dynamics and are described by power flow equations. Including load buses does not change the stability properties of the closed-loop system, though the argument is slightly more complicated~\cite{zhao2014design,mallada2017optimal}.

\subsection{Failure Propagation}\label{section:propagation}
In full generality, the failure model and the control strategy we introduce later apply to both generator failures and line failures, but, to simplify the presentation, in this paper we focus only on the latter.
We describe the cascading failure process by keeping track of the set of outaged lines $\calB(\n)\subset \calE$ over stages $\n\in\set{1,2,\ldots, N}$ at steady state. Following a line outage, we assume that the system evolves on a fast timescale according to the frequency dynamics
\eqref{eqn:swing_and_network_dynamics} during the transient phase. When it eventually converges to a closed-loop equilibrium
, overloaded lines are tripped and the cycle repeats, as illustrated in Fig.~\ref{fig:model_interplay}. 

More specifically, for each stage $\n\in \set{1,2,\ldots, N}$, the system evolves according to the dynamics \eqref{eqn:swing_and_network_dynamics} on the topology $\calG(\n)$, and converges to an equilibrium 
point $x^*(\n)=(\theta^*(\n), \omega^*(\n), d^*(\n), f^*(\n))$ that solves an optimization problem over $\calG(\n)=\left(\calN,\calE\setminus \calB(\n)\right)$.
If at equilibrium all branch flows $f^*(\n)$ are within the corresponding line limits, then $x^*(\n)$ is a secure operating point and the cascade stops. Otherwise, let $\calF(\n)$ be the subset of lines whose branch flows exceed the corresponding line limits. The lines in $\calF(\n)$ operate above their limits at steady state, so we assume they trip at the end of stage $\n$ and set $\calB(\n+1)=\calB(\n)\cup \calF(\n)$. Line overloads during the transient phase before the system converges to $x^*(\n)$ are considered tolerable because the transient dynamics in \eqref{eqn:swing_and_network_dynamics} does not last long enough to overheat a line \cite{zhao2016unified} (spanning only seconds to a few minutes). This process then repeats for the subsequent stages.
We remark that our model captures a simplified line trip dynamic where, after an initiating transmission failure (consisting of one or more line trips), the next set of (one or more) line trips does not occur until the frequency control dynamics has converged. In particular our model does not capture the finer thermal dynamics that determines the timing of line trips based on line currents and ambient temperatures.

The crux of our failure propagation model lies in the interplay between the slow-timescale line tripping process and the fast-timescale dynamics. By explicitly modeling fast-timescale frequency control dynamics as part of the cascading process, our model offers more flexibility in controller design. Different choices of $d(t)$ induce different cascading failure processes. For instance, as shown in Appendix~\ref{section:previous_model}, if we adopt droop control for $d(t)$, the failure model in~\cite{soltan2015analysis, yan2015cascading} (where injections do not change after a non-cut failure, and a bridge failure impact the injections following a certain power balancing rule $\calR$) can be readily recovered. As another example, if AGC is adopted for $d(t)$, the cascading process will unfold in a way where injections and line flows are changed even after a non-cut failure. Since traditional AGC does not enforce line limits (congestion is managed on a slow timescale), some lines may carry flows above their thermal limits and are tripped subsequently.

This integrated model offers an additional means to mitigate cascading failures through a better design of the frequency control mechanism on a fast timescale. Our proposed approach leverages this extra freedom and adopts a recent frequency control approach known as Unified Controller (UC) for $d(t)$. In contrast to traditional AGC, UC drives the closed-loop system to an equilibrium that respects line limits whenever possible. We will show in Section \ref{section:non_critical} that it will in fact localize the impact even when the outaged lines disconnect the network.

\begin{figure}[t]
\vspace{.1cm}
\centering
\includegraphics[width=0.4\textwidth]{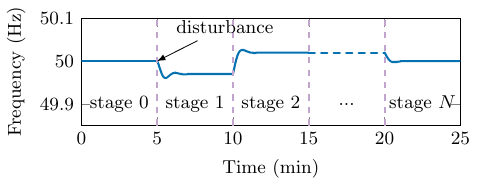}
\caption{An illustration of the failure propagation model.}
\label{fig:model_interplay}
\vspace{-.3cm}
\end{figure}

\section{The Bridge-block Decomposition and The Unified Controller}\label{section:block_controller}
The bridge-block decomposition and the unified controller have recently emerged as two important tools for grid reliability \cite{part1,part2,zhao2016unified}. The two concepts operate on different timescales to improve the power system robustness: the bridge-block decomposition aims to localize the failure propagation, while the unified controller aims to stabilize a disturbed system. In this section, we review these concepts and elaborate on how they can be integrated as a novel control framework for failure localization and mitigation.

\subsection{Bridge-block Decomposition}
Given a power network $\calG=(\calN,\calE)$, a \emph{partition} of $\calG$ is defined as a finite collection $\calP=\{\calN_1, \calN_2, \cdots, \calN_k\}$ of nonempty and disjoint subsets of $\calN$ such that $\bigcup_{i=1}^k \calN_i=\calN$. For a partition $\calP$, each edge can be classified as either a \emph{tie-line} if the two endpoints belong to different subsets of $\calN$ or an \emph{internal line} otherwise.  

We define an equivalence relation on $\calN$ such that two nodes are in the same equivalence class if and only if there are two edge-disjoint paths connecting them. For this specific partition, the tie-lines connecting different components are exactly the \textit{bridges} (cut-edges) of the graph. We thus refer to this partition as \emph{bridge-block decomposition} $\calP^{\text{BB}}$ of the power network (see Fig.~\ref{fig:bridge-block} for an example). 

It is shown in \cite{part1} that each graph has a unique bridge-block decomposition, which can be found in linear time. In particular, the bridge-block decomposition encodes rich information on failure propagation.
\begin{figure}[t]
\vspace{.1cm}
\centering
\includegraphics[width=0.2\textwidth]{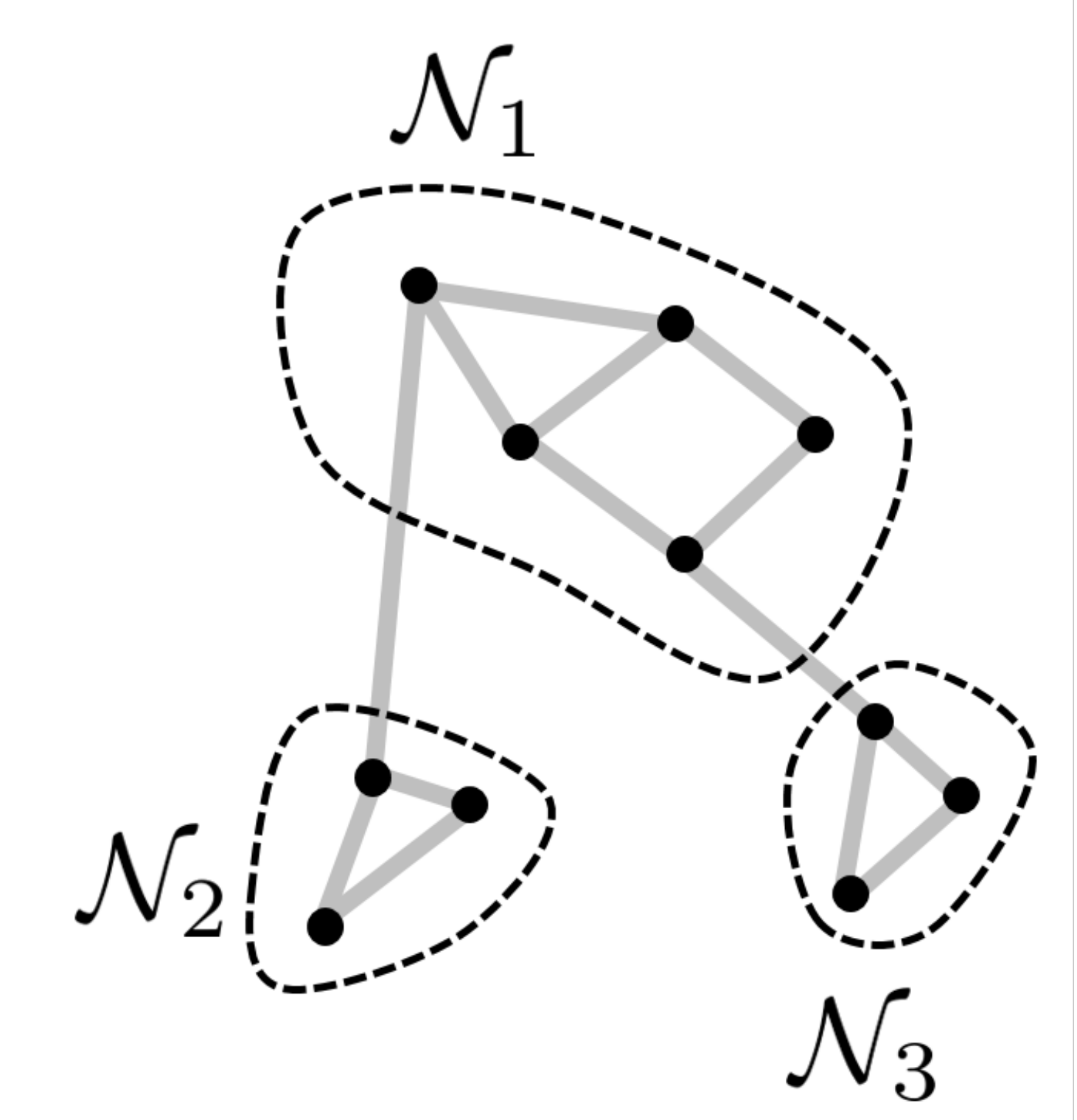}
\caption{Bridge-block decomposition of a graph.}
\label{fig:bridge-block}
\vspace{-.3cm}
\end{figure}

\subsection{Unified Controller (UC)}
UC is a control approach recently proposed in the frequency regulation literature \cite{zhao2014design,li2016connecting,zhao2016unified,mallada2017optimal,zhao2018distributed}. Compared to classical droop control or AGC \cite{bergen2009power}, UC simultaneously integrates primary control, secondary control, and congestion management on a fast timescale.
The key feature of UC is that the \emph{closed-loop} equilibrium of \eqref{eqn:swing_and_network_dynamics} under UC solves the following optimization problem on the \emph{post}-contingency network:
\begin{subequations}\label{eqn:uc_olc}
\begin{IEEEeqnarray}{ll}
\min_{\theta, \omega, d, f} \quad& \sum_{j\in\calN}c_j(d_j) \label{eqn:uc_obj}\\
\hspace{.2cm}\text{s.t.} & \omega=0,\label{eqn:uc_secondary}\\
&r + d - Cf = 0, \label{eqn:uc_balance}\\
& f = BC^T\theta,  \label{eqn:dcflow}\\
& ECf = 0,\label{eqn:uc_ace}\\
& \ul{f}_{e}\le f_{e}\le \ol{f}_{e}, \quad e\in\calE, \label{eqn:line_limit}\\
& \ul{d}_{j}\le d_j\le \ol{d}_j, \quad j \in\calN, \label{eqn:control_limit}
\end{IEEEeqnarray}
\end{subequations}
where $c_j(\cdot)$'s are associated cost functions that penalize deviations from the last optimal dispatch  (and hence attain minimum at $d_j = 0$), \eqref{eqn:uc_secondary} ensures secondary frequency regulation is achieved, \eqref{eqn:uc_balance} guarantees power balance at each bus,  \eqref{eqn:dcflow} is the DC power flow equation, \eqref{eqn:uc_ace} enforces zero area control error \cite{bergen2009power}, \eqref{eqn:line_limit} and \eqref{eqn:control_limit} are the flow and control limits. The matrix $E$ encodes balancing area information as follows. Given a partition $\calP^{\text{UC}}= \set{\calN_1,\calN_2,\cdots, \calN_k}$ of $\calG$ that specifies the balancing areas for secondary frequency control, $E\in\set{0,1}^{|\calP^{\text{UC}}|\times n}$ is defined by $E_{lj}=1$ if bus $j$ is in balancing area $\calN_l$ and $E_{lj}=0$ otherwise. As a result, the $l$-th row of $ECf=0$ ensures that the branch flow deviations on the tie-lines connected to balancing area $\calN_l$ sum to zero.

UC is designed so that its controller dynamics, combined with the system dynamics \eqref{eqn:swing_and_network_dynamics}, form a variant of projected primal-dual algorithms to solve \eqref{eqn:uc_olc}. It is shown in \cite{zhao2014design,li2016connecting,zhao2016unified,mallada2017optimal,zhao2018distributed} that when the optimization problem \eqref{eqn:uc_olc} is feasible, under mild assumptions, the closed-loop equilibrium under UC is globally asymptotically stable and it is an optimal point of \eqref{eqn:uc_olc}. 
Such an optimal point is unique (up to a constant shift of $\theta$) if the cost functions $c_j(\cdot)$ are strictly convex. This means that, after a (cut or non-cut) failure, the post-contingency system is driven by UC to an optimal solution of \eqref{eqn:uc_olc} (under appropriate assumptions). We refer the readers to \cite{zhao2014design,li2016connecting,zhao2016unified,mallada2017optimal,zhao2018distributed} for specific controller designs and their analysis.

\subsection{Connecting UC and the Bridge-Block Decomposition}
We have introduced two partitions of a power network: the bridge-block decomposition $\calP^{\text{BB}}$ and the balancing area partition $\calP^{\text{UC}}$, which in general are different from each other. However, when they do coincide, the underlying power grid inherits analytical properties from both bridge-block decomposition and UC, making the system particularly robust against failures. Our proposed control strategy leverages precisely this feature, as we present in Section \ref{section:control_strategy}. 

In practice, the balancing areas over which UC operates are usually connected by multiple tie-lines in a mesh structure. However, in order to align with the bridge-block decomposition, we may have to switch off
a few tie-lines of $\calP^{\text{UC}}$. The selection of these tie-lines can be systematically optimized, e.g., to minimize line congestion or inter-area flows on the resulting network; see \cite{zocca2019spectral} for more details. We henceforth assume that $\calP^{\text{BB}}=\calP^{\text{UC}}$. 
We refer to such a network as the \emph{tree-partitioned} network since the balancing areas are connected in a tree structure prescribed by its bridge-block decomposition. 

\begin{defn} \label{def:critical}
	Given a cascading failure process described by $\calB(\n)$, with $\n\in\set{1,2, \ldots, N}$, the set $\calB(1)$ is said to be its \textbf{\emph{initial failure}}.
	An initial failure $\calB(1)$ is said to be \textbf{\emph{critical}} if the UC optimization \eqref{eqn:uc_olc} is infeasible over $\calG(1):=(\calN, \calE\bs\calB(1))$, or \textbf{\emph{non-critical}} otherwise.
\end{defn}

To formally state our localization result, we define the following concept to clarify the precise meaning of an area being ``local'' with respect to an initial failure.

\begin{defn}
	Given an initial failure $\calB(1)$, we say that a tree-partitioned balancing area $\calN_l$ is \textbf{\emph{associated}} with $\calB(1)$ if there exists an edge $e=(i,j)\in\calB(1)$ such that either $i\in\calN_l$ or $j\in\calN_l$.
\end{defn}

As we discuss below, our control strategy possesses a strong localization property for both non-critical and critical failures
in the sense that only the operation of the associated areas are adjusted whenever possible.

\section{Proposed Control Strategy: Summary}\label{section:control_strategy}

Our strategy consists of two phases: a planning phase in which tree-partitioned networks are created and an operation phase during which UC actively monitors and autonomously reacts to line failures during its operation.

\subsection{Planning Phase: Align Bridge-blocks and Balancing Areas}
Each balancing area of a multi-area power network is managed by an independent system operator (ISO). 
Although these areas exchange power with each other as prescribed by economic dispatch, their operations are relatively independent.
This is usually achieved via the zero area control error constraint in secondary frequency control \cite{bergen2009power}, which is enforced by UC with \eqref{eqn:uc_ace}. As mentioned in previous sections, such balancing areas typically do not form bridge-blocks, as redundant lines are believed to be critical in maintaining $N-1$ security of grid \cite{bergen2009power,bienstock2007integer,hines2007controlling}.

We propose to create bridge-blocks whose components coincide with the balancing areas over which UC operates. This can be done by switching off a small subset of the tie-lines so that areas are connected in a tree structure. The switching actions only need to be carried out in the planning phase, as line failures that occur during the operating phase do not affect the bridge-blocks already in place.\footnote{In fact, line failures can lead to a ``finer'' bridge-block decomposition, as more bridge blocks are potentially created when lines are removed from service.} 
It is interesting to notice that when the subset of lines to switch off is chosen carefully, the tree-partitioned network not only localizes the impact of line failures, but can also \emph{improves} overall reliability. This seemingly counter-intuitive phenomenon is illustrated by our case studies in Section \ref{section:case_study_n_1}.

We remark that it is natural to match bridge-blocks to balancing areas in transmission systems. First, there are typically small number of tie-lines between areas. Switching off those tie-lines is less likely to cause severe congestion to the grid. Second, each balancing area is supposed to balance its own power and the primary goal of failure mitigation is to prevent, as much as possible, a failure from impact other areas. This reconciles with the area control error (ACE) in traditional secondary frequency control.

In practice, the switching actions should be performed judiciously without jeopardizing the network operations. Any line switching action leads to power redistribution on remaining lines and may cause overload.  It is thus crucial to assess the impact of each set of switching actions and avoid creating any congested line. It should be noted that switching off lines may not necessarily increase the transmission cost~\cite{fisher2008optimal}. Moreover, the system congestion level may even decrease when carefully choosing tie-lines to switch off~\cite{zocca2019spectral}. We refer interested readers to~\cite{zocca2019spectral} for an efficient algorithm to optimize such a line selection.

\subsection{Operating Phase: Extending Unified Controller}
Once the tree-partitioned areas are formed, the power network operates under UC as a closed-loop system and responds to disturbances such as line failures or loss of generator/load in an autonomous manner. Unlike traditional secondary frequency control, UC is distributed and reacts continuously to any disturbances as part of normal operation. There is no explicit termination of UC primal-dual dynamics.
We assume that UC can detect line failures in real-time and thus react to the post-contingency network accordingly. In normal conditions where the system disturbances caused by line failures are small, UC always drives the power network back to an equilibrium point that can be interpreted as an optimal solution of \eqref{eqn:uc_olc}. This is the case, for instance, when non-critical failures (see Definition \ref{def:critical}) happen, and therefore such failures are always successfully mitigated.

However, in extreme scenarios where a major disturbance (e.g., a critical failure) happens, the optimization problem \eqref{eqn:uc_olc} can be infeasible. In other words, it is physically impossible for UC to achieve all of its control objectives after such a disturbance. 
This makes UC unstable (see Proposition \ref{prop:dual_to_infinity}) and may lead to successive failures.
There is therefore a need to extend the version of UC proposed in \cite{zhao2014design,li2016connecting,zhao2016unified,mallada2017optimal,zhao2018distributed} with two features: (a) a detection mechanism that monitors the system state and detects critical failures promptly; 
and (b) a constraint-lifting mechanism that responds to critical failures by proactively relaxing certain constraints of \eqref{eqn:uc_olc} to ensure system stability can be reached at minimal cost. 

Our technical results in Section \ref{section:critical_detection} suggest a way to implement both components as part of the normal operation of UC. System operators can prioritize different balancing areas by specifying the sequence of constraints to lift in response to extreme events. This allows the non-associated areas to be progressively involved and coordinated in a systematic fashion when mitigating critical failures. We discuss some potential schemes in Section \ref{section:critical_contraint_lift}.
\subsection{Guaranteed Mitigation and Localization}
We show in detail in Sections \ref{section:non_critical} and \ref{section:critical} that the proposed strategy provides strong guarantees in the mitigation and localization of both non-critical and critical failures. Specifically, it ensures that the cascading process is always stopped (a) after a non-critical failure by the associated areas, and the operating points of non-associated areas are not impacted in equilibrium or (b) after a critical failure when constraints in \eqref{eqn:uc_olc} are lifted in a progressive manner specified by the system operator. Thus the proposed strategy can always prevent successive failures, while localizing the impact of the initial failures as much as possible.

We remark that our strategy can maintain the $N-1$ security standard even though the balancing areas are connected in a tree structure. Unlike other classical approaches where failures that disconnect the network tend to incur more severe impact, our strategy can mitigate such failures as much as possible by autonomously adjusting the injections to rebalance and stabilize the system in each of the surviving component. In fact, our fast-timescale control can significantly improve the grid reliability in $N-k$ sense as we show in the numerical experiments in Section~\ref{section:case_studies}.

\section{Localizing Non-critical Failures} 
\label{section:non_critical}
In this section, we consider non-critical failures, as defined in Section \ref{section:block_controller}, and prove that such failures are always fully mitigated within the associated balancing areas.

We first characterize how the system operating point shifts in response to such failures. Recall that if an initial failure $\calB(1)$ is non-critical, the UC optimization \eqref{eqn:uc_olc} is feasible and thus the new operating point $x^*(1):=(\theta^*(1), \omega^*(1), d^*(1), f^*(1))$ satisfies all the constraints in \eqref{eqn:uc_olc}. In particular, none of the line limits is violated at $x^*(1)$ by \eqref{eqn:line_limit}, i.e.~$x^*(1)$ is a secure operating point and the cascade stops ($\calF(1) = \emptyset$). Moreover the power flows on bridges remain unchanged in equilibrium from their pre-contingency values, as the next result says. 

\begin{lemma}\label{lemma:zero_branch_deviation}
Given a non-critical initial failure $\calB(1)$, the new operating point $x^*(1)$ prescribed by the UC satisfies $f^*_e(1)=0$ for every bridge $e$ of the network.
\end{lemma}

\begin{proof}
To simplify the notation, we drop the stage index $(1)$ from $x^*$ and denote $x^*=(\theta^*,\omega^*,d^*,f^*)$. Given a bridge $e=(j_1,j_2)$ of $\calG$, removing $e$ partitions $\calG$ into two connected components, say $\calC_1$ and $\calC_2$. Without loss of generality, assume $j_1\in\calC_1$ and $j_2\in\calC_2$. For an area $\calN_v$ from the partition $\calP$, we say $\calN_v$ is within $\calC_1$ if for any $j\in\calN_v$ we have $j\in\calC_1$. It is easy to check from the definition of a tree-partitioned network that any area $\calN_v$ from $\calP$ is either within $\calC_1$ or within $\calC_2$, and that $e$ is the only edge in $\calG$ that has one endpoint in $\calC_1$ and the other endpoint in $\calC_2$.

Let $\calP'$ be the subset of areas within $\calC_1$ from $\calP$, and let $\bff{1}_{\calP'}\in\set{0,1}^{\abs{\calP}}$ be its characteristic vector (that is, the $l$-th component of $\bff{1}_{\calP'}$ is $1$ if $\calN_l\in\calP'$ and $0$ otherwise). Given two buses $i$ and $j$, we denote $i\goesto j$ if $(i,j)\in\calE$ and $j\goesto i$ if $(j,i)\in\calE$.

Note that \eqref{eqn:uc_ace} ensures the injections are balanced for all areas. We can thus sum over all areas within $\calC_1$:
\begin{IEEEeqnarray}{rCl}
0&=&\bff{1}_{\calP'}^TECf^*\nonumber
\end{IEEEeqnarray}
Following the definition of matrices $E, C$ and the above notations, we can rewrite the above equations as:
\begin{IEEEeqnarray}{rCl}
0 &=&\sum_{l:\calN_l\in\calP'}\sum_{i\in\calN_l}\big (\sum_{j:j\goesto i}f^*_{ji} - \sum_{j:i\goesto j} f^*_{ij}  \big)\nonumber\\
&=&\sum_{i:i\in\calC_1}\big (\sum_{j:j\goesto i}f^*_{ji} - \sum_{j:i\goesto j} f^*_{ij}\big )
\label{eqn:uc_sum_injection}
\end{IEEEeqnarray}
Mathematically, $\big (\sum_{j:j\goesto i}f^*_{ji} - \sum_{j:i\goesto j} f^*_{ij}\big ) +r_i+d_i^*=0$ is the flow conservation for node $i$. Therefore, the right hand side of equation~\eqref{eqn:uc_sum_injection} calculates the summation of injections over all nodes within the connected component $\calC_1$, which should be zero enforced by \eqref{eqn:uc_ace}.

Now, consider the node $j_1$ and the bridge $e=(j_1,j_2)$. For all other transmission lines $(i,j)$ within component $\calC_1$, the flow will be counted twice (for nodes $i$ and $j$ respectively) with opposite direction in~\eqref{eqn:uc_sum_injection} and thus gets cancelled out. Only the flow of bridge $e$ is counted once for node $j_1$. Therefore, we conclude that the flow on the bridge must be zero, i.e. $f_e^*=0$. Since the bridge $e$ is arbitrary, we have thus proved the desired result.

\end{proof}
This lemma shows that tree-partitioned areas enable UC to achieve more than what it was originally designed for in \cite{zhao2014design,li2016connecting,zhao2016unified,mallada2017optimal,zhao2018distributed}: the extended UC not only enforces zero area control errors through \eqref{eqn:uc_ace}, it also guarantees zero flow deviations on all bridges. 

The following proposition is another result of this type, which clarifies how the tree-partitioned network induces a localization property under UC. 

\begin{prop}\label{prop:localizability}
Assume $c_j(\cdot)$ is strictly convex and achieves its minimum at $d_j = 0$ for all $j\in\calN$. Given a non-critical initial failure $\calB(1)$, if an area $\calN_l$ is not associated with $\calB(1)$, then at equilibrium $x^*(1)$ we have $d^*_j(1)=0$ for all $j\in\calN_l$.
\end{prop}

\begin{proof}
To simplify the notation, we drop the stage index from the equilibrium $x^*$ and write $x^*=(\theta^*,\omega^*, d^*,f^*)$ and $p^* = r+d^*$.

First, we construct a different point $\tilde{x}^*=(\tilde{\theta}^*, \tilde\omega^*, \tilde{d}^*, \tilde{f}^*)$ by changing certain entries of $x^*$ within a non-associated area $\calN_l$ as follows: 
(a) replace $d^*_j$ with $\tilde{d}_j^*=0$ for all $j\in\calN_l$; (b) replace $f^*_e$ with $\tilde{f}^*_e=0$ for $e\in\calE$ that have both endpoints in $\calN_l$; and (c) replace $\theta^*$ by a solution 
$\tilde{\theta}^* = L^\dag \tilde p^*$ obtained from solving the DC power flow equations with injections $\tilde{p}^* = r+\tilde{d}^*$.
All other entries of $x^*$ remain unchanged in $\tilde x^*$.
Since $c_j(\cdot)$ attains its minimum at $d_j = 0$, $\tilde{x}^*$ achieves at most the same objective value \eqref{eqn:uc_obj} as $x^*$. Thus $\tilde{x}^*$ must be an optimal point of \eqref{eqn:uc_olc}, provided it is feasible.

When the cost functions $c_j(\cdot)$ are strictly convex, the optimal solution to  \eqref{eqn:uc_olc} is unique in $d^*$ and $f^*$ ($\theta^*$ is also unique up to an arbitrary reference angle). As a result, if the constructed point $\tilde{x}^*$ is feasible, We can then conclude that $\tilde{x}^*=x^*$ (up to an arbitrary reference angle). 

We now prove the feasibility of $\tilde{x}^*$. The construction of $\tilde{x}^*$ ensures that \eqref{eqn:uc_ace}\eqref{eqn:line_limit}\eqref{eqn:control_limit} are satisfied. 
If we can show that $\tilde{f}^*=BC^T\tilde{\theta}^*$, then since $\tilde{\theta}^*$ is obtained by solving the DC power flow equations from $CBC^T\tilde{\theta}^*=\tilde{p}^*$, constraints \eqref{eqn:uc_balance} and \eqref{eqn:dcflow} are also satisfied, proving the feasibility of 
$\tilde x^*$.
It thus suffices to show $\tilde{f}^*=BC^T\tilde{\theta}^*$. To do so, we first establish the following lemma:
\begin{lemma}\label{lemma:uc_injection_vanish}
For any tree-partitioned area $\calN_z$ in $\calP$, we have
$
\sum_{j\in\calN_z} {p^*_j}=\sum_{j\in\calN_z} {\tilde{p}^*_j}=0
$.
\end{lemma}
\begin{proof}
Let $\bff{1}_{\calN_z}\in\R^{\abs{\calN}}$ be the characteristic vector of $\calN_z$, that is, the $j$-th component of $\bff{1}_{\calN_z}$ is $1$ if $j\in\calN_z$ and $0$ otherwise. Summing \eqref{eqn:uc_balance} over $j\in\calN_z$, we have:
$
\sum_{j\in\calN_z}{p^*_j}=\bff{1}_{\calN_z}^TCf=(ECf)_z=0,
$
where $(ECf)_z$ is the $z$-th row of $ECf$. 
If $\calN_z = \calN_l$, we have $\tilde{p}^*_j=0$ for $j\in\calN_l$ by construction and hence $\sum_{j\in\calN_z} {\tilde{p}^*_j}=0$. For $\calN_z\neq \calN_l$, we have $\tilde{p}^*_j=p^*_j$ for any $j\in\calN_z$ by construction. Thus, $\sum_{j\in\calN_z} {\tilde{p}^*_j}=0$, completing the proof.
\end{proof}

Consider now an area $\calN_w$ that is different from $\calN_l$. 
In this case, we do not change the injections from $x^*$ when constructing $\tilde{x}^*$, thus $p_j^*-\tilde{p}_j^*=0$ 
for all $j\in\calN_w$. From Lemma \ref{lemma:uc_injection_vanish}, we see that $\sum_{j\in\calN_z}\paren{p_j^*-\tilde{p}_j^*}=0$ for all $z$. 
Since both $(p^*, \theta^*)$ and $(\tilde p^*, \tilde \theta^*)$ satisfy the DC power flow equations, we have
$
CBC^T (\theta^* - \tilde{\theta}^*) \ = \ p^* - \tilde{p}^*.
$
\begin{lemma}\label{lemma:laplace_solvability}
Let $\calP=\set{\calN_1,\calN_2,\cdots, \calN_k}$ be the tree-partitioned areas of $\calG$ and consider a vector $p\in\R^{n}$ such that $p_j=0$ for all $j\in\calN_1$ and $\sum_{j\in\calN_z} p_j=0$ for $z\neq 1$. Then the Laplacian equation $CBC^T \theta = p$
is solvable, and any solution $\theta$ satisfies $\theta_i=\theta_j$ for all $i,j\in\ol{\calN}_1$, where $\ol{\calN}_1:=\set{j:\exists \, i\in\calN_1 \text{ s.t. } (i,j)\in\calE \text{ or } (j,i) \in \calE}$.
\end{lemma}

\begin{proof}
It is well-known (see \cite{guo2017monotonicity} for instance) that the Laplacian matrix $L:=CBC^T$ of a connected graph $\calG=(\calN,\calE)$ has rank $\abs{\calN} - 1$, and $L\theta=p$ is solvable if and only if $\bff{1}^T p=0$, where $\bff{1}$ is the vector with proper dimension that consists of ones. Moreover, the kernel of $L$ is given by $\spann(\bff{1})$.

If $\calN_1$ is the only area in $\calP$, then $p=0$ since $p_j=0$ for all $j\in\calN_1$. We thus know the solution space to $L\theta=p$ is exactly the kernel of $L$, and the desired result holds.

If $\calN_1$ is not the only area in $\calP$, we can show that $f_e=0$ for every bridge $e$ connecting different areas. To see this, we start from a leaf area $\calN_z$ of the tree-partitioned network. Since $\sum_{j\in\calN_z} p_j=0$, we have that the only bridge connecting the leaf area will carry zero flow. Therefore, the area $\calN_z$ is decoupled from other areas and the Laplacian equation can be decomposed. We can iteratively use the above statement and decompose the Laplacian equation with respect to each area. With previous result, we know $\theta_i=\theta_j$ for all $i,j\in\calN_1$ since $p_j=0$ for all $j\in\calN_1$. Together with the fact that bridges carry zero flow, we have $\theta_i=\theta_j$ for all $i,j\in \ol{\calN}_1$. 
\end{proof}

By Lemma~\ref{lemma:laplace_solvability}, we then have $\theta^*_j-\tilde{\theta}^*_j$ is a constant over $\ol{\calN}_l$, and thus $\tilde{\theta}^*_i - \tilde{\theta}^*_j \ = \ \theta^*_i -\theta^*_j $
for all $i, j\in\ol{\calN}_l$. This in particular implies
$
\tilde{f}^*_e \ = \ f^*_e \  = \ B_e(\theta^*_i -\theta^*_j) \ = \  B_e(\tilde{\theta}^*_i - \tilde{\theta}^*_j)
$
for all $e=(i,j)$ such that $i\in\calN_w$ or $j\in\calN_w$.

Finally, consider the area $\calN_l$. We have $\tilde{p}^*_j=0$ by construction. 
From Lemma \ref{lemma:uc_injection_vanish} we have $\sum_{j\in\calN_z} \tilde{p}_j^* =0$ for all $z$. Since $CBC^T\tilde{\theta}^*=\tilde{p}^*$, we know $\tilde{\theta}^*_i=\tilde{\theta}^*_j$ for all $i,j\in\ol{\calN}_l$. This implies that for any edge $e=(i,j)$ within $\calN_l$, we have
$
\tilde{f}^*_e=0=B_e(\tilde{\theta}^*_i-\tilde{\theta}^*_j).
$
and therefore  $\tilde{f}^*_e=B_e(\tilde{\theta}^*_i-\tilde{\theta}^*_j)$ for all $e\in\calE$, concluding the proof.
\end{proof}

Proposition \ref{prop:localizability} reveals that with the proposed control strategy, after a non-critical failure, the injections and power flows in non-associated areas remain unchanged at equilibrium, even though they fluctuate during transient according to \eqref{eqn:swing_and_network_dynamics}.
Our control scheme guarantees that non-critical failures in a balancing area do not impact the operations of other areas, achieving stronger balancing area independence than that ensured by zero area control error alone.

Furthermore, traditional control strategies usually treat  failures that disconnects the system differently, i.e. the post-contingency injections normally stay constant, but can be changed if the system are disconnected due to the line failures~\cite{soltan2015analysis}. Under our scheme, failures that disconnect the system are treated in exactly the same way as failures that do not, provided that they are non-critical. Moreover, the impact of a failure that disconnects the system is localized and properly mitigated to the associated areas as well. This is in stark contrast with the global and severe impact of a bridge failure~\cite{part2} and is the key benefit of integrating UC with the bridge-block decomposition.


\section{Controlling Critical Failures}\label{section:critical}
We now consider the case where the initial failure is critical. This may happen when a major generator or transmission line is disconnected from the grid.

\subsection{Unified Controller under Critical Failures}\label{section:critical_detection}
Since UC is a concept that has emerged from the frequency regulation literature, 
the underlying optimization \eqref{eqn:uc_olc} is always assumed to be feasible in existing studies \cite{zhao2014design,li2016connecting,zhao2016unified,mallada2017optimal,zhao2018distributed}. 
As such, little is known about the behavior of UC if this assumption is violated when a critical failure happens. 
We now characterize the limiting behavior of UC in this setting.

In order to do this, we first formulate the exact controller dynamics of UC. Unfortunately, there is no standard way to do this as multiple designs of UC have been proposed in the literature \cite{zhao2014design,li2016connecting,zhao2016unified,mallada2017optimal,zhao2018distributed}, each with its own strengths and weaknesses. Nevertheless, all of the proposed controller designs are (approximately) projected primal-dual algorithms for the optimization problem \eqref{eqn:uc_olc} satisfying two assumptions that we now state. Let $\lambda_i$, for $i\in\set{1,2,\cdots, n + 3m+\abs{\calP^{\text{UC}}}}$, be the dual variables corresponding to the constraints \eqref{eqn:uc_balance}-\eqref{eqn:line_limit}.

\textbf{UC1}: For all $j\in\calN$, $\ul{d}_j\le d_j(t)\le \ol{d}_j$ is satisfied for all $t$. This is achieved either via a projection operator that maps $d_j(t)$ to this interval or by requiring the cost function $c_j(\cdot)$ to approach infinity near these boundaries.

\textbf{UC2}: The primal variables $f, \theta$ and the dual variables $\lambda_i$ are updated by a primal-dual algorithm\footnote{We do not consider the specific variants of primal-dual algorithms that are proposed in different designs of UC, since the standard primal-dual algorithm is often a good approximation.} to solve \eqref{eqn:uc_olc}.

\begin{prop} \label{prop:dual_to_infinity}
Assume UC1 and UC2 hold. If \eqref{eqn:uc_olc} is infeasible, then there exists a dual variable $\lambda_i$ such that 
$\limsup_{t\goesto\infty}\abs{\lambda_i(t)} = \infty$.
\end{prop}
\begin{proof}
First, collect in the vector $x= (\theta, \omega, d, f) \in\R^{3n+m}$ 
all the decision variables of the UC optimization \eqref{eqn:uc_olc} and rewrite it in a more standard form as
\begin{equation}\label{eqn:pd_opt}
    \min_{\ul{d} \le d \le \ol{d}} \quad  c(d) \quad \text{s.t.} \quad Ax \le g, \  Cx = h
\end{equation}
where $A, C, g, h$ are matrices (vectors) of proper dimensions that can be recovered from the full formulation in \eqref{eqn:uc_olc}. Let $\lambda_1,\lambda_2$ be the corresponding dual variables, and set $\lambda:=[\lambda_1;\lambda_2]$ ($[\cdot;\cdot]$ here means matrix concatenation as a column). We can then write the Lagrangian for \eqref{eqn:pd_opt} as 
$$\calL(x, \lambda) = c(d)  + \lam_1^T (Ax-g) + \lam_2^T (Cx-h).$$
By the assumption UC2, we know that:
\begin{equation*}\label{eqn:pd_alg}
\dot{\lam}_1 = [Ax-g]_{\lam_1}^+, \quad \dot{\lam}_2 = Cx-h
\end{equation*}
with the projection operator $[\cdot]_{a}^+$ defined component-wise by $\left([x]_{a}^+\right)_i=x_i$ if $x_i>0$ or $(a)_i >0$, and $\left([x]_{a}^+\right)_i=0$ otherwise.
Consider two closed convex sets $S_1 = \{x~|~ Ax\le g, Cx = h\}$ and $S_2 = \{x~|~\ul{d} \le d \le \ol{d}\}$. If the optimization \eqref{eqn:uc_olc} is infeasible, then $S_1\cap S_2=\emptyset$. As a result, we can find a hyper-plane that separates $S_1$ and $S_2$: more specifically, there exists $q\in \R^{3n+m}, q_0 \in \R$ such that 
$$
q^Tx > q_0, \, \forall \, x \in S_1 \, \text{ and } \, q^Tx \le q_0, \, \forall \, x \in S_2.
$$
This fact then implies that the system 
$
\set{Ax\le g, \quad Cx=h, \quad q^Tx\le q_0}
$
is not solvable. By Farkas' Lemma, we can thus find vectors $w_1, w_2, w_3$ such that $w_1 \ge 0, w_3 \ge 0$ (the inequality is component-wise), $A^T w_1 + C^T w_2 + q w_3 = 0$, and $g^T w_1 + h^T w_2 +  q_0w_3 = - \eps< 0$.
Define $z = [w_1; w_2]$. We then see that under the UC controller, we have for any $t$:
\begin{subequations}\label{eqn:pd_infea}
	\begin{IEEEeqnarray}{rCl}
		z^T \dot{\lambda}(t)&= &w_1^T [Ax(t)-g]_\lam^+ + w_2^T (Cx(t)-h) \nonumber \\
		&\ge& w_1^T [Ax(t)-g]_\lam^+ + w_2^T (Cx(t)-h)\nonumber\\
		&& + \ w_3(q^T x(t) - q_0) \label{eqn:primal_in_range}\\
		&\ge & w_1^T (Ax(t)-g) + w_2^T (Cx(t)-h)\nonumber\\
		&& + \ w_3(q^T x(t) - q_0) \label{eqn:remove_bracket} \\
		&=&\paren{A^Tw_1  + C^Tw_2 + qw_3}x(t)\nonumber\\
		&& - \paren{w_1^T g + w_2^T h + w_3 q_0} = 0+\eps>0\nonumber
	\end{IEEEeqnarray}
\end{subequations}
where \eqref{eqn:primal_in_range} follows from $w_3\ge 0$ and  assumption UC1, which ensures $x(t)\in S_2$ and thus $q^Tx(t)-q_0 \le 0$, and \eqref{eqn:remove_bracket} comes from $w_1\ge 0$ and the fact that $[x]_\lambda^+\ge x$ for all $x$ (the inequality is component-wise). Consequently,$z^T \lambda(t) -z^T\lambda(0)> \eps t$
and thus
$
\lim_{t\goesto \infty} z^T\lambda(t) = \infty.
$
Finally, by noting
$
\lim_{t\goesto \infty} z^T\lambda(t) \le w_1^T\limsup_{t\goesto\infty}\abs{\lambda_1(t)} + \abs{w_2}^T\limsup_{t\goesto\infty}\abs{\lambda_2(t)},
$
the desired result follows. 
\end{proof}
Proposition \ref{prop:dual_to_infinity} implies that, after a critical failure, UC cannot drive the system to a proper and safe operating point. This type of instability suggests a way to detect critical failures. Specifically, since Proposition \ref{prop:dual_to_infinity} guarantees that at least one dual variable becomes arbitrarily large in UC operation when \eqref{eqn:uc_olc} is infeasible, we can set thresholds for the dual variables and raise an infeasibility warning if any of them is exceeded. By doing so, critical failures can be detected in a distributed fashion during the normal operations of UC. 

Note that the dual variables for non-critical failures are usually bounded in practice unless certain degeneracy in \eqref{eqn:uc_olc} happens. This, however, does not impact the stability of UC. That being said, non-critical failures may also cause relatively large dual variable values in transient states. The choice of detection thresholds thus inevitably involves trade-offs. Specifically, tighter thresholds allow critical failures to be detected more promptly so that the system only experiences a short period of instability.  On the other hand, tighter thresholds also lead to a larger false alarm rate for non-critical failures. As we will show in the next section, when a failure is detected being critical, certain protection mechanisms will be involved so that the system is stabilized with potential cost of non-local response and load loss. Therefore, these thresholds should be chosen carefully by the operator in accordance to specific system parameters and application scenarios.

\subsection{Constraint Lifting as a Remedy}\label{section:critical_contraint_lift}
In the event of a critical failure, it is impossible for UC to simultaneously achieve all of its control objectives and constraints. This can lead to instability and thus successive failures. We can mitigate this by progressively lifting certain constraints from UC in two different ways without compromising the basic objective of stabilizing the system:
\begin{itemize}
	\item The zero area control error constraints \eqref{eqn:uc_ace} between specific pairs of balancing areas can be lifted. This means the controller now involves more balancing areas in failure mitigation.
	\item Loads can be shedded, which is reflected  in \eqref{eqn:uc_olc} by enlarging the range $[\ul{d}_j, \ol{d}_j]$ for corresponding load buses.
\end{itemize}


By iteratively lifting these two types of constraints, we can guarantee the feasibility of \eqref{eqn:uc_olc} and ensure that the system converges to a stable point that is free from successive failures. This, however, comes with the cost of potential load loss, and thus must be carried out judiciously. 
The iterative relaxation procedure can follow predetermined rules specified by the system operator to prioritize different objectives. As an example, one can minimize load loss by relaxing possibly all area control error constraints before relaxing injection bounds on load buses.
This will utilize all the contingency and regulation reserves globally across all areas to meet demand and shed load only as a last resort. In contrast, if the localization of failure impact should be prioritized, the operator can choose to first lift load injection bounds in the associated areas and then progressively lift area control error constraints to get more balancing areas involved.

In practice, these two types of constraint lifting can be combined and implemented iteratively, following the predetermined rules from the system operator. Specifically, the system operator may assess the node importance of the grid and design a hierarchical constraint lifting procedure, which specifies the order for involving balancing areas and the allowance for injection adjustment. Once the procedure is specified, the real-time response of UC can be implemented accordingly by changing the dynamics of dual updates; see~\cite{mallada2017optimal} for the detailed design of UC dynamics.

\section{Case Studies}\label{section:case_studies}
 In this section, we evaluate the performance of the proposed control strategy on the IEEE 118-bus and IEEE 2736-bus (the Polish network) test systems, with respect to $N-k$ security standard and localization performance under different levels of system congestion.

\subsection{$N-k$ Security under Different Congestion Levels}\label{section:case_study_n_1}
We first focus on system robustness with respect to $N-k$ security standard on the IEEE 118-bus system. This test network has two balancing areas shown as Area 1 and Area 2 in Fig.~\ref{fig:IEEE_118}. To form a tree-partitioned network, three lines $(15,33)$, $(19,34)$, and $(23,24)$ 
are switched off, obtaining what we henceforth refer to as the \textit{revised} network.

\begin{figure}[t]
\centering
\includegraphics[width=.9\columnwidth]{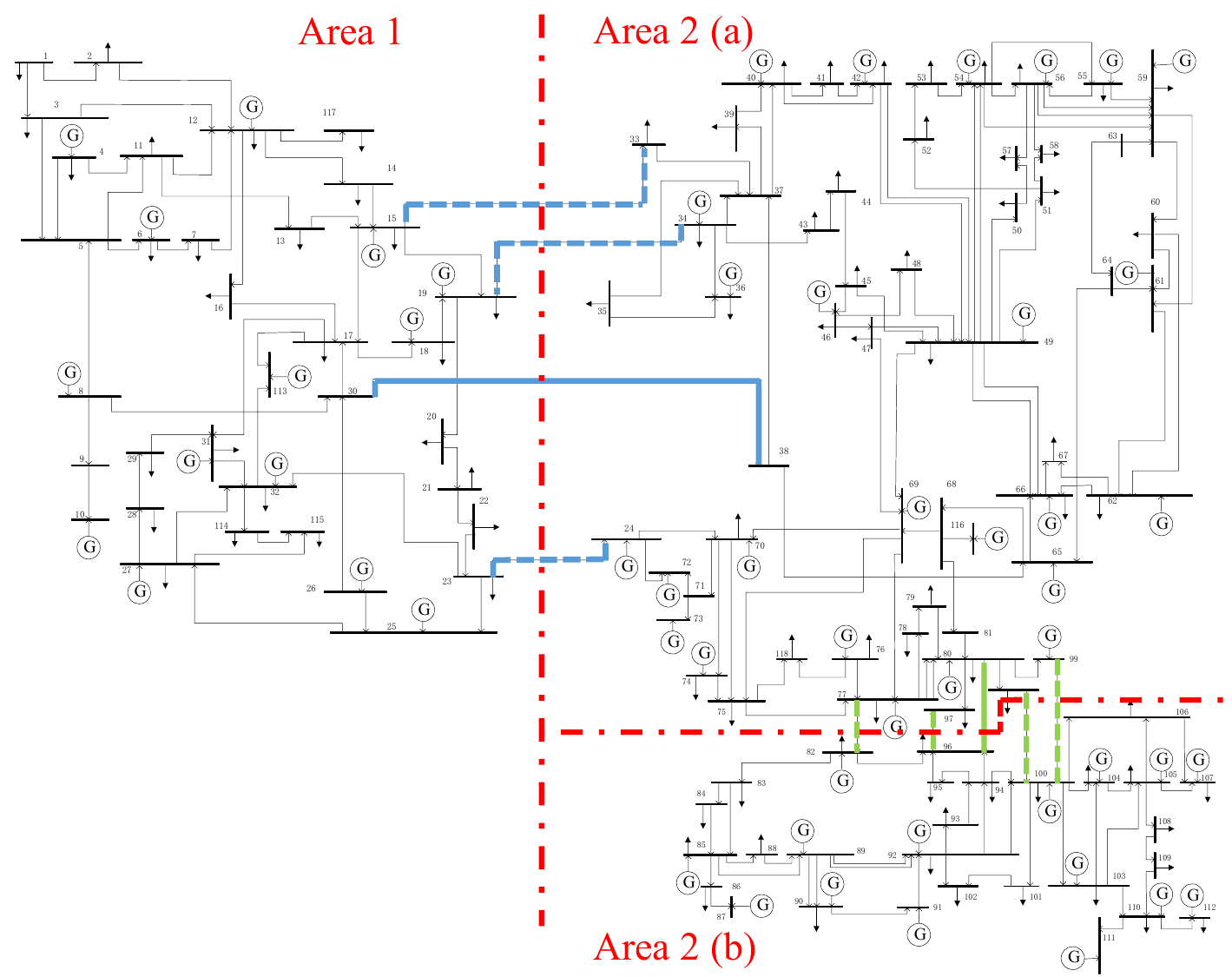}
\caption{One line diagram of the IEEE 118-bus test network.
	}\label{fig:IEEE_118}
\end{figure}

We compare UC on the tree-partitioned revised network, as specified by our proposal, and classical AGC on the \emph{original} network. 
UC is modeled by the optimization problem \eqref{eqn:uc_olc} 
and AGC is modeled by \eqref{eqn:uc_olc} without the line limits \eqref{eqn:line_limit}. 
A failure scenario is said to be \emph{vulnerable} if the initial failure leads to successive failures or loss of load. 
To compare the performance between our proposed approach and AGC, we collect statistics on (a) vulnerable scenarios as a percentage of the total simulated scenarios, and (b) load loss rate (LLR) which is defined as the ratio of the total load loss to the original total demand.  
Note that we do not perform time-domain simulations, but directly compute the equilibrium point of the closed-loop systems under UC and AGC by solving the coressponding optimization respectively.
This reduces the computational complexity and allows us to systematically evaluate our proposed control method in a wide range of simulation settings. Moreover, we assume that the detection threshold is set relatively large {so that the failures can be accurately classified into critical and non-critical failures. Therefore, constraint lifting only happens for critical failures. }

The failure scenarios are created as follows. First, we generate a variety of load injections 
by adding random perturbations to the nominal load profile from \cite{zimmerman2011matpower} and then solve the DC OPF to obtain the corresponding generator operating points {over the actual IEEE test networks. For a fair comparison, we use the same generation and load operating points over the revised network and calculate the pre-contingency flows with DC power flow equations.} 
Second, we sample over the collection of all subsets that consist of $k$ transmission lines of the IEEE 118-bus test network. Finally, for each sampled subset of $k$ lines, we remove all lines in this subset as initial failure and simulate the cascading process thus triggered. 
Specifically, for $k=1,2,3$ initial line failures, we generate 100, 15, and 15 load profiles and further compute the optimal generation dispatch by DC OPF. For each load profile, we iterate over every single transmission line failure, and sample 3,000 and 5,000 failure scenarios for $k=2,3$ line failures respectively. In total, our simulations cover the cases $k=1,2,3$ with roughly 138,600 failure scenarios, as summarized in Table.~\ref{table:sampling_rate}.

\begin{table}[h]
	\def\arraystretch{1.3}
	\centering{
		\caption{Simulation setup for $N-k$ security evaluation.}\label{table:sampling_rate}
		\begin{tabular}{|c|c|c|c|}
			\hline
			Case & $k=1$ & $k=2$ & $k=3$\\
			\hline
			\# of Load Profiles & 100 & 15 & 15 \\
			\hline
			\# of Sampled Failures & 186 & 3000 & 5000\\
			\hline
			Total Scenarios & 18600 & 45000 & 75000\\
			\hline
		\end{tabular}
	}
\end{table}

Fig.~\ref{fig:n_k}(a) shows the average, minimum, and maximum percentage of vulnerable scenarios across all sampled failure scenarios, while Fig.~\ref{fig:n_k}(b) plots the complementary cumulative distribution (CCDF) of the load loss rates. The simulation results show that the proposed control incurs both substantially fewer vulnerable scenarios and much less loss of load in all cases compared to AGC. This difference is particularly pronounced when multiple lines are tripped simultaneously ($k=2,3$). We highlight that in our simulations, UC operates over the tree-partitioned network (while AGC operates over the original network) in which some of the tie-lines are switched off and hence some transfer capacity is removed from the system. Moreover,
the newly created bridge $(30, 38)$ in the tree-partitioned network is never vulnerable under the proposed control in all the scenarios
we have studied.

\begin{figure}[t]
\centering
\subfloat[Average fraction of vulnerable scenarios for $k=1,2,3$. \label{fig:vulnerable_lines_wrt_k}] 
{\includegraphics[width=0.35\textwidth]{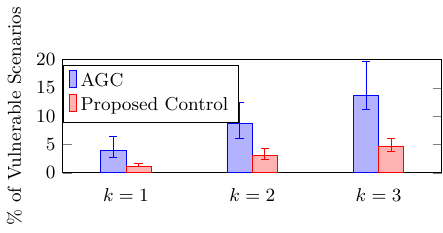}}
\hfill
\subfloat[CCDF for load loss rate for $k=1,2,3$. \label{fig:load_loss_nk}] {\includegraphics[width=0.35\textwidth]{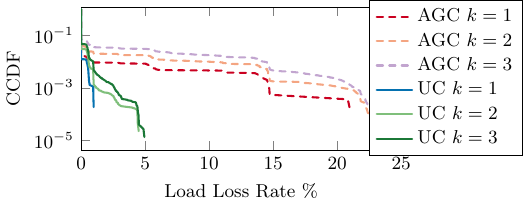}}
\caption{System robustness in terms of the $N-k$ security standard.}
\label{fig:n_k}
\end{figure}
\begin{figure}[t]
\centering
\subfloat[Average fraction of vulnerable scenarios for different congestion levels. \label{fig:vulnerable_lines}] {\includegraphics[width=0.35\textwidth]{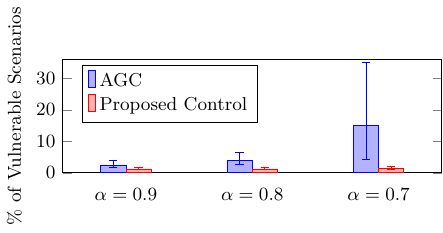}}
\hfill
\subfloat[CCDF for load loss rate for different congestion levels. \label{fig:load_loss_n1}] {\includegraphics[width=0.35\textwidth]{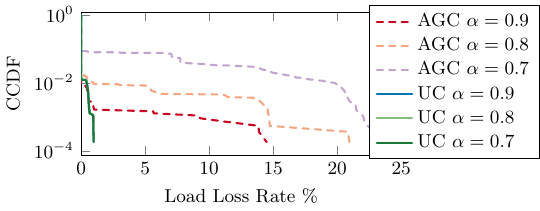}}
\caption{System robustness under different levels of congestion obtained scaling line capacities by the factor $\alpha$.}
\label{fig:n_congestion}
\end{figure}
\begin{figure}[ht]
\centering
\subfloat[CCDF for generator response on IEEE 118-bus network. \label{fig:gen_response_118}] {\includegraphics[width=0.35\textwidth]{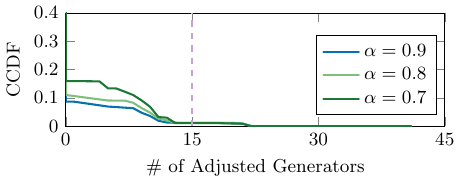}}
\hfill
\subfloat[CCDF for generator response on IEEE Polish network. \label{fig:gen_response_2736}] {\includegraphics[width=0.35\textwidth]{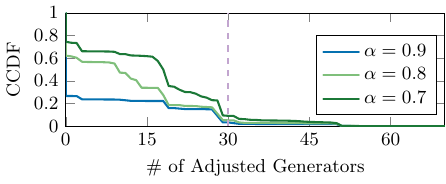}}
\caption{Generator response over IEEE 118-bus and Polish network.}
\label{fig:n_gen}
\end{figure}

We then illustrate the improvement of the proposed approach over AGC under different congestion levels. To do so, we scale down the line capacities to $\alpha=0.9, 0.8, 0.7$ of the base values and collect statistics for all single line initial failures ($k=1$). Our results are summarized in Fig.~\ref{fig:n_congestion}, which again show that the proposed approach significantly outperforms AGC in all scenarios, especially those in which the system is congested. Again in all these scenarios, the bridge  is never vulnerable under the proposed control.

{We remark that UC outperforms AGC even if both control methods operate over the tree-partitioned network since AGC does not enforce line limits. There are only a few failure scenarios with a lower load loss rate under AGC. However, those failures always propagate through multiple stages and require a significantly longer time to re-stabilize the system.}


\subsection{Localized Failure Mitigation}\label{section:localized_exp}

In this subsection, we consider a specific constraint lifting rule that progressively involves other areas by relaxing area control error constraints only if local load shedding within the associated areas is not enough to make problem  \eqref{eqn:uc_olc} feasible. This rule prioritizes localization of the initial failure. By implementing it, we show that the proposed control strategy can localize cascading failures within the associated areas with negligible load loss. The experiments are carried out over two networks: (a) a finer tree-partitioned version of the IEEE 118-bus test network, and (b) the much larger-scale Polish network, consisting of 2736 buses and 3504 lines.

For the IEEE 118-bus test network, we switch off 4 additional lines, $(77,82)$, $(96,97)$, $(98,100)$, and $(99,100)$, which refines the bridge-blocks used in the previous subsection since it further decomposes Area 2 into two balancing areas (as shown in Fig.~\ref{fig:IEEE_118}). The generator capacities are scaled down by 60\% so that the total generation reserve is roughly 20\%. We create different congestion levels by scaling the line capacities according to a factor $\alpha=0.9,0.8,0.7$ and iterate over all single transmission lines as initial failures. The injections are the same as that for the $N-1$ test in the previous subsection.

The statistics on the fraction of vulnerable scenarios and load loss rates (LLRs) for this experiment are summarized in Table~\ref{table:case118}. We observe that the proposed control strategy never incurs more than 2.21\% LLR across all tested injections and congestion levels. Furthermore, for this specific network, the proposed approach localizes \emph{all} failures to the associated areas, i.e., the tie-line constraints are never lifted. This localization phenomenon can more clearly be noticed in Fig.~\ref{fig:n_gen}(a), which shows the CCDF of the number of generators whose operating points are adjusted in response to the initial failures. The majority of failures lead to operating point adjustments on less than 15 generators, which is roughly the number of generators within a single balancing area. The small portion of failures that impact more than 15 generators are bridge failures, which by definition have two associated areas and thus more associated generators.

\begin{table}[t]
\def\arraystretch{1.4}
\centering{
\caption{Statistics on failure localization over the IEEE 118-bus test network.}\label{table:case118}
\begin{tabular}{|c|c|c|c|}
\hline
 Line Capacity & $\alpha=0.9$ & $\alpha=0.8$ & $\alpha=0.7$\\
\hline
Avg. \% of Vul. Sce. & 3.53 & 3.68 & 3.82 \\
\hline
Avg. (Max.) LLR(\%) & 0.55 (1.06) & 0.56 (2.17) & 0.59 (2.21) \\
\hline
\end{tabular}
}
\end{table}


For the Polish network, we switch off 78 transmission lines from the original network, creating a tree-partitioned network with 4 areas of 1430, 818, 359 and 129 buses respectively. Similar to the setup for the IEEE-118 test network, the generation capacities are scaled properly so that the total generation reserve is roughly 20\%, and the line capacities are scaled down to $\alpha=0.9, 0.8, 0.7$ to create different congestion levels. We then iterate over all single line failures and the statistics from our experiments are summarized in Table~\ref{table:case2736}. Our results show that for this test network, more than 86\% of the single line failures can be mitigated locally within a single area for all congestion levels. In addition, the worst case LLR is roughly 3\% across all simulated scenarios, with an average that is no higher than 0.07\%. Similar to the IEEE 118-bus test network, the number of generators whose operating points are adjusted by the proposed control strategy in response to the failures is small, as shown in  Fig.~\ref{fig:n_gen}(b), confirming failure localization.

\begin{table}[t]
\def\arraystretch{1.3}
\centering{
\caption{Statistics on failure localization over the Polish network.}\label{table:case2736}
\begin{tabular}{|p{3cm}|c|c|c|}
\hline
\hspace{.7cm}Line Capacity & $\alpha=0.9$ & $\alpha=0.8$ & $\alpha=0.7$\\
\hline
Scenarios Mitigated with one Area (\%) & 92.39 & 88.63 & 86.91 \\
\hline
Scenarios Mitigated with 2-3 Areas (\%) & 6.44 & 9.48 & 10.40 \\
\hline
Scenarios Mitigated with All Areas (\%) & 1.17 & 1.89 & 2.69 \\
\hline
Avg. (Max.) LLR(\%) & 0.05 (2.93) & 0.05 (2.94) & 0.07 (3.24) \\
\hline
Avg. \# of Gen. Adj. & 6.52 & 11.66 & 16.37 \\
\hline
\end{tabular}
}
\end{table}

\section{Conclusion}\label{section:conclusion}
In this paper, we have proposed a complementary approach to grid reliability by integrating the network bridge-block decomposition and the unified controller for frequency regulation to achieve fast-timescale failure control. It provides strong analytical guarantees of both the localization and mitigation of failures. Our case studies on the IEEE 118-bus and 2736-bus test systems show that the proposed control scheme can greatly improve overall reliability compared to the current practice. In particular, the new control prevents successive failures from happening while localizing the impacts of initial failures. When load shedding is inevitable, the proposed strategy incurs significantly less load loss.%



\begin{appendices}
\section{Recovering Previous Models}\label{section:previous_model}

The dynamic model \eqref{eqn:swing_and_network_dynamics} in Section \ref{section:failure_model} models secondary frequency control where the frequency deviations $\omega(t)$ are driven to zero.  When we focus on controllers that only achieve primary frequency control, the equilibrium frequency $\omega^*$ may be nonzero. That is, as the system converges in this sense, the phase angles $\theta^*(t)$ do not necessarily stay at a constant value, but may change in constant rate over time. In such context, we can modify \eqref{eqn:swing_and_network_dynamics} as follows to describe primary frequency dynamics:
\begin{subequations}\label{eqn:swing_and_network_dynamics_primary}
	\begin{IEEEeqnarray}{rCll}
		M_j\dot{\omega}_j &=& r_j + d_j - D_j\omega_j - \sum_{e\in\calE}C_{je}f_e,&\quad j\in\calN \label{eqn:swing_dynamics_p}\\
		{f}_{ij}&=&B_{ij}(\theta_i-\theta_j),\quad& (i,j)\in\calE. \label{eqn:network_flow_dynamics_p}
	\end{IEEEeqnarray}
\end{subequations}

By relaxing the requirement on $\omega^*=0$ at equilibrium, the above model enables extra freedom in the choice of $d_j$. We now show that by
using the classical droop control \cite{bergen2009power} as the dynamics for $d_j$'s in \eqref{eqn:swing_and_network_dynamics_primary}, the cascading failure models from~\cite{part2} and previous literature such as \cite{soltan2015analysis, yan2015cascading} can be readily recovered. Indeed, as shown in \cite{zhao2016unified}, the \emph{closed-loop} equilibrium of \eqref{eqn:swing_and_network_dynamics_primary} under droop control is the unique\footnote{The equilibrium is unique up to an arbitrary reference phase angle.} optimal solution to the following optimization on the post-contingency network:
\begin{subequations}\label{eqn:droop_olc}
\begin{IEEEeqnarray}{ll}
\min_{\theta, \omega, d, f} \quad & \sum_{j\in\calN}\frac{d_j^2}{2Z_j}+\frac{D_jw_j^2}{2}\label{eqn:droop_obj}\\
 \hspace{.2cm}\text{s.t.}& r - d - D\omega = Cf\label{eqn:droop_balance}\\
	& f - BC^T \theta = 0\label{eqn:droop_branch}\\
	&\ul{p}_j\le r_j - d_j \le \ol{p}_j, \quad j\in\calN, \label{eqn:droop_control_limits}
\end{IEEEeqnarray}
\end{subequations}
where $Z_j$'s are the generators' participation factors \cite{bergen2009power}. By plugging \eqref{eqn:droop_branch} into \eqref{eqn:droop_balance}, it is easy to check that any feasible point $x=(\theta, \omega, d, f)$ of \eqref{eqn:droop_olc} satisfies $\sum_j r_j=\sum_j (d_j + D_j\omega_j)$. 
Cauchy-Schwarz inequality then implies that
\begin{IEEEeqnarray*}{rCl}
\Big( \sum_{j\in\calN} r_j \Big)^2&=& \Big[ \sum_{j\in\calN} \paren{d_j+D_j\omega_j}\Big]^2\\
&\le & \sum_{j\in\calN}\Big (\frac{d_j^2}{2Z_j}+\frac{D_j\omega_j^2}{2} \Big)\sum_{j\in\calN}\paren{2Z_j+2D_j},
\end{IEEEeqnarray*}
for which  equality holds if and only if 
\begin{equation}\label{eqn:droop_optimal_point}
d_j=\frac{Z_j}{\sum_j \paren{Z_j+D_j}}\sum_j r_j, \quad \omega_j=\frac{\sum_j r_j}{\sum_j \paren{Z_j+D_j}}.
\end{equation}
Therefore, if the control limits \eqref{eqn:droop_control_limits} are not active, \eqref{eqn:droop_optimal_point} is always satisfied at the optimal point $x^*=(\theta^*,\omega^*, d^*, f^*)$.

Now, consider a line $e$ being tripped from the transmission network $\calG$, and for simplicity assume the control limits \eqref{eqn:droop_control_limits} are not active. If $e$ is a bridge, the tripping of $e$ results in two islands of $\calG$, say $\calD_1$ and $\calD_2$, and two optimization problems \eqref{eqn:droop_olc} correspondingly. For $l=1,2$, $\sum_{j\in\calD_l} r_j$ represents the total net power imbalance in $\calD_l$, and therefore \eqref{eqn:droop_optimal_point} implies that droop control adjusts the system injections so that the power imbalance is distributed to all generators proportional to their participation factors in both $\calD_1$ and $\calD_2$. If $e=(i,j)$ is not a bridge, denoting the original flow on $e$ before it is tripped as $f_e$, then $r_i=f_e$, $r_j=-f_e$ and $r_k=0$ otherwise. As a result, we have $\sum_{j\in\calN}r_j = 0$ in this case and thus \eqref{eqn:droop_optimal_point} implies the system operating point remains unchanged in equilibrium, i.e., $d_j=\omega_j=0,\forall j\in \calN$. Moreover, one can show that this still holds when \eqref{eqn:droop_control_limits} is active with a more involved analysis on the KKT conditions of \eqref{eqn:droop_olc}.
This droop control mechanism recovers the failure propagation model in~\cite{part2} and underlies some of previous results in the literature on cascading failures in power systems \cite{bernstein2014power,soltan2015analysis, yan2015cascading}. 
\end{appendices}

\bibliographystyle{IEEEtran}
\bibliography{biblio}

\end{document}